\documentclass[11pt,a4paper]{article}

\usepackage[in]{fullpage}
\usepackage{natbib}
\usepackage{apalike}
\usepackage{amssymb}
\usepackage{amsmath}
\usepackage{amsthm}
\usepackage{graphicx}
\usepackage{floatrow}
\usepackage[hidelinks]{hyperref}
\usepackage{sidecap}
\usepackage{txfonts}
\usepackage{authblk}

\newcommand{\aj}{Astron. J.}

\newcommand{\apj}{Astrophys. J.}
\newcommand{\aap}{Astron. Astrophys.}

\newcommand{\planss}{Planetary and Space Science}
\newcommand{\icarus}{Icarus}
\newcommand{\mnras}{MNRAS}

\newtheorem{theorem}{Theorem}

\theoremstyle{remark}
\newtheorem*{remark}{Remark}

\let\OLDthebibliography\thebibliography
\renewcommand\thebibliography[1]{
  \OLDthebibliography{#1}
  \setlength{\parskip}{0pt}
  \setlength{\itemsep}{5pt plus 0.3ex}
}

\begin{document}

\title{Asteroid models from generalised projections: Essential facts for asteroid modellers and geometric inverse problem solvers}

\author[1]{Mikko Kaasalainen}
\author[2]{Josef \v{D}urech}
\affil[1]{Tampere University, Tampere, Finland}
\affil[2]{Astronomical Institute, Faculty of Mathematics and Physics, Charles University, Prague, Czech~Republic, e-mail: durech@sirrah.troja.mff.cuni.cz}
\renewcommand\Affilfont{\itshape\normalsize}
\date{\normalsize\today}

\maketitle

\begin{abstract}
\noindent
We present a review of the problem of asteroid shape and spin reconstruction from generalised projections; i.e., from lightcurves, disk-resolved images, occultation silhouettes, radar range-Doppler data, and interferometry. The aim of this text is to summarize all important mathematical facts and proofs related to this inverse problem, to describe their implications to observers and modellers, and to provide the reader with all relevant references.
\end{abstract}

\section{Introduction}
The title combines astronomy and mathematics in three ways: while our application is asteroid modelling, generalised projections are the mathematical concept that defines the data we use, and ``models from'' means that we are dealing with mathematical inverse problems. Make no mistake: while the data are obtained by near-automatic instruments and the analysis of the data can now be carried out largely by software tools made for this purpose, it's not just machines. Some basic mathematical understanding of the problem is always required as well, even if the calculations, methods, and principles are hidden in the software. Creating a model that fits the data is not the same as understanding and solving the inverse problem. This is why we have written this brief history and review of asteroid models from inversion techniques that are demonstrably exact, well-founded and well-established, explaining the main points and leaving details to the references. There are hundreds of papers related to this problem field by now, so one can easily lose track of the basic principles. The main point is that this is a mathematical rather than merely computational problem. This text concerns the shape and spin (and some surface characteristics) of asteroids; mathematical modelling of the insides of asteroids from in-situ radio data is a newer topic, described in, e.g., \cite{Pur.Kaa:13, Pur.Kaa:16}.

By now the theory and practice of inverse problems of {\em generalised projections} (GPs) is well established, so asteroid modelling rests on a solid mathematical basis. GPs are, in short, integrals over parts of the surface of a three-dimensional body that correspond to various types of observational data \citep{Kaa.Lam:06}. They generalise the concept of a projection in a given direction (silhouette) to various forms such as photographic or radar images. In instrumental sense, quanta from the surface are integrated over time to form a quantitative reading. What the integral of a GP looks like depends on the instrument. Here we discuss three classes of GPs in three sections: (i) lightcurves, (ii) images and silhouettes, and (iii) radar and interferometry.

\section{Lightcurves}
\subsection{Shape}
In its simplest form, a GP is the apparent brightness of a body when illuminated from some direction and viewed from another. The observable at different geometries as a function of time is called a {\em lightcurve}; the corresponding GP integral is over the visible and illuminated part of the surface, taking into account its light-scattering properties.

If the brightness GP is known for several geometries, can the shape of the body be inferred from these? This is a fundamental mathematical problem. Or, rather, an {\em inverse problem}: computing the brightness of a known surface is quite straightforward, whereas the inverse direction is something else. Is there a unique answer? Even if yes, what happens to the result if the data are noisy or sparse or we are not sure how the surface scatters light?

The first step towards the solution of the inverse problem was given in \cite{Rus:1906} \citep[for a modern version, see][]{Kaa.ea:92}: if the viewing and illumination directions are the same, the inverse problem does not have a unique answer. About half of the information needed for reconstructing the shape is inevitably lost, so there is no chance of a decent solution. Russell assumed geometric scattering; i.e., brightness directly proportional to the area of the silhouette, which is roughly the case for dark surfaces \citep[in][it was shown that the same nonuniqueness applies to, e.g., diffuse bright surfaces but, in fact, concerning the other half of the shape information]{Kaa.ea:92}. Since Russell's geometry was assumed to be a good approximation of realistic asteroid observations, the result stopped the study of the inverse problem for almost a century. Approximations by fixed shapes such as ellipsoids were supposed to be the only thing one can get out of brightness data \citep{Mag.ea:89}.

The inverse problem was shown to be uniquely solvable in \cite{Kaa:90} and \cite{Kaa.ea:92} (interestingly, group-theoretical tools employed in quantum mechanics turned out to be necessary for the proof). The uniqueness requires separate viewing and illumination directions; obviously, the larger the separation (phase angle) the better. The uniqueness proof requires an assumption of a convex surface to facilitate general analytical solutions;\footnote{The fundamental nonuniqueness (independent of the illumination phase angle) and some uniqueness properties of nonconvex solutions are proven in \cite{Vii.ea:17} and \cite{Kaa:19a}.} the scattering behaviour, however, is not especially restricted (as long as the intrinsic darkness or albedo is uniform). What is more, the inverse problem is actually rather stable: moderate errors in the data or the assumed scattering cause only moderate rather than drastic changes in the inferred shape. This is called Minkowski stability, since the solution requires the solution of the Minkowski problem that is very stable \citep{Lam:93}. One version of the problem is: if the areas of the facets of a convex polyhedron are known, is the shape unique? Minkowski proved that it is, and provably robust shape reconstruction procedures are described in, e.g., \cite{Kaa.ea:92, Kaa.ea:01}, \cite{Lam:93}, and \cite{Lam.Kaa:01}. 

Once the uniqueness was assured, the next step was to see what happens with real data. As is typical for many cases, not all the methods used in the proof were directly suitable for actual inversion. Both convex and nonconvex shapes were used in numerical inversion in \cite{Kaa.ea:01}; the former for the proven uniqueness, and the latter for generality despite the ambiguities. A robust practical procedure, based on nonlinear optimisation to find the rotational properties of the target and solving for the facet areas of a polyhedron representation and then applying the Minkowski procedure, is given in \cite{Kaa.ea:01} and available as software at Database of Asteroid Models from Inversion Techniques \citep[DAMIT,][]{Dur.ea:10}.\footnote{\url{http://astro.troja.mff.cuni.cz/projects/damit/}} Indeed, the convex version turned out to be Minkowski stable in practice, and even less than ten good lightcurves from at least a few different geometries provide enough data for a unique model.  In fact, the explanatory level of convex inversion turned out to be too good in a sense: it seems that practically all lightcurves of known asteroids can be explained well with convex shapes. There are only few exceptions known out of over a thousand asteroids. So, while the solution was reassuringly stable, it turned out that there is a considerable limit to the detail one can infer from lightcurves.

The information limit was estimated in \cite{Dur.Kaa:03}. Even if there are quite large nonconvex features on the surface, they usually have no unique signature for most realistic geometries: the solar phase angle should be several tens of degrees of arc (only near-Earth asteroids can be observed at such geometries) before a convex shape cannot explain the data. Thus, in most cases, nonconvex models simply cannot be more than an artist's vision even though they look more realistic than a convex shape. This is a pity, since nonconvex inversion is actually easier technically than the convex one -- see the nonconvex methods and models in \cite{Kaa.ea:01, Kaa.ea:04, Vii.ea:15a}, also available at DAMIT.\footnote{However, for probing and sampling the shape solutions, nonconvex codes are not practical since each code depends on its choice of shape representation (support), regularisation for smoothness etc., and the corresponding setups -- see Sect.~\ref{sec:disk-resolved_data}, \cite{Vii.Kaa:14, Vii.ea:15a}. There is no comprehensively coverable nonconvex shape class in the same sense that there is a convex class; it is approximated by a collection of quite separate classes. Each class and setup chooses its best-fit solution.} The features of a nonconvex solution may look temptingly similar to those of probe-visited targets, but this is because the shape representation usually interpolates between surface points in a natural manner (e.g., by rendering a flat region seen in a convex solution as a smooth valley). This, of course, does not mean that the feature is based on any actual information if the convex model fits the data as well.

In mathematical sense, the convex version of the inverse problem is not very ill-posed: the answer is unique, and all that is needed to assure a stable solution is a positivity constraint and moderate regularisation to ensure a feasibility constraint. Computation and convergence are also very fast. On the other hand, the nonconvex version has inherent ambiguities \citep{Vii.ea:17, Kaa:19a}, requires regularisation for stability, and usually cannot be guaranteed in any sense to be correct in its features and details. Of course, few asteroids are exactly convex, so a convex solution inherently warns us of the same information limit: in most cases, we can only infer coarse global-scale shape properties. Further corroboration is provided by ground truth, the accurate shape models from space-probe data (asteroids such as Eros, Gaspra, and Itokawa). Lightcurves computed from these models do not fit the observed ones as well as those from convex inversion models. This is essentially due to systematic errors in data and the scattering model that is never known accurately \citep{Kaa.Dur:07}; in any case, it emphasises the fundamental resolution limit of brightness data. 

Many of the concepts above can be conveniently examined in the simplistic two-dimensional context \citep{Ost.Con:84, Kaa:16}. This is not useful for real-world problems as such, but the importance of the illumination phase angle and the essence of Minkowski stability can be explained very compactly. Indeed, a main point in interpreting brightness mostly due to the shape rather than albedo changes is not just that this is realistic (as seen from probe images) but that the solution is much more stable. Significant albedo changes can be inferred from brightness data, in
which case this can be modelled as a spot in a 3-D shape model \citep[e.g., Psyche,][]{Vii.ea:18}. Such models are not unique; for example, one can always make a spot area sharper and brighter and/or the opposite end more blunt and darker. The existence and estimate of albedo variegation is deduced in the complete modelling process. Strong asymmetry of a single lightcurve may imply an albedo spot only if the data are from zero solar phase angle.

\subsubsection{Implications for observers}
\label{sec:implications}

\paragraph{Solar phase angle} A convex shape model can be uniquely reconstructed only if the viewing and illumination geometries are not the same, i.e., when the solar phase angle is not zero. In practice, lightcurve inversion works well up to the distance of Jupiter Trojans with heliocentric distances of about 5\,AU, where the maximum phase angle is $\sim\!12^\circ$. For more distant populations of small bodies observed from Earth, the solar phase angle is always too small to enable shape reconstruction from lightcurves.

\paragraph{Relative vs. absolute photometry} For the purpose of shape reconstruction, absolute photometry is always better than just relative lightcurves with unknown zero point. However, most lightcurve observations come in form of relative data. This is no problem for the inversion -- scale factors for individual lightcurves (shifts in magnitude scale) are free nuisance parameters of the optimization. The disadvantage of relative lightcurves is that the shape uncertainty along the rotation $z$-axis is large. An extreme example would be observation of straight-line lightcurves in which case we cannot tell how flattened the spheroid is. Also with long-period rotation when the full lightcurve cannot be observed in a single night, relative lightcurves provide poor information about the amplitude.

\paragraph{Colours} If we have calibrated observations in different colours, the colour indices can be easily obtained from the inversion as ratios between scale factors for lightcurves in different colours.

\paragraph{Observation noise} The random noise is not really an issue as long as it is not large -- very accurate lightcurves cannot give a more detailed model -- cf. our photometric laboratory asteroid in \cite{KaaS.ea:05} with lightcurve accuracy of order 0.001. The laboratory asteroid shows the ground truth and proves that a model with certainly wrong light-scattering properties nevertheless is a good shape+spin one (Minkowski stability) and can fit the data exactly (so the accurate data do not increase resolution). One should keep in mind that there is always a shape/albedo ambiguity present, so any attempts to reconstruct shape details are (apart from other effects) inevitably limited by the assumption of light-scattering properties being constant over the surface.

\subsection{Spin}
\label{sec:spin}

So far we have discussed the shape solution. With asteroids, one must simultaneously solve for the spin; i.e., the period, the pole ecliptic latitude $\beta$ and longitude $\lambda$. Precessing asteroids or time-dependent periods are a straightforward extension of this \citep{Kaa:01, Kaa.ea:07}, so we only consider the majority class of fixed-spin targets. Many of these have orbits close to the Earth's orbital plane that place them in the domain of the 180-degree spin longitude ambiguity \citep{Kaa.Lam:06}. It means that for geometry restricted to one plane, the coordinates of the Sun and Earth in the asteroid frame are exactly the same for the original asteroid and such the is mirrored along its equatorial plane, its pole longitude is $\lambda + 180^\circ$, and its initial rotation is $\phi_0 + 180^\circ$. In practice, this affects the spin latitude and shape estimates as well, but usually not very much. This fundamental ambiguity concerns all data integrated over the whole surface, regardless of the wavelength. It applies to any process history or sequence of disk-integrated data, regardless of the orbital positions, if it occurs at the planar geometry. Thus, e.g., thermal infrared cannot resolve this ambiguity.

This is in fact the only nontrivial spin ambiguity. For completeness, let us prove that the spin solution is unique for non-planar viewing geometries:

\begin{theorem}
 For extensive geometry coverage, there is only one shape and spin combination that matches the data (assuming a non-spheroidal shape and nonzero spin).
\end{theorem}

\begin{proof}
  First we note that the planar $\lambda$+180-ambiguity is the only condition of invariance of reproduced data with a simple transformation of the essentially invariant shape. This is seen from the three rotation matrices corresponding to each spin angle \citep{Kaa.Lam:06}. The sinusoidal elements of the matrices do not allow further simple mimicking transformations. In the trivial case of zero spin (a constant rotation-phase matrix), there is an infinite variety of rotated but otherwise similar shapes that match the data. 

  In the uniqueness proof of shape, we assumed infinitely dense, accurate, and full coverage of observations since that is the only way of having an exact proof without haggling about the setup, accuracy, and number of observations (those belong to the realm of stability, not uniqueness which always assumes ideal conditions). But with rotation, the concept of ideal coverage is not so easily defined independent of the measurement setup. Let us first assume that we have a system that records the brightness from all directions around the rotating target in an instant (compared to the rotation rate). The information from one instant is sufficient to determine the shape separately. Then we need only one more instant to infer the spin properties (assuming the instants are within one half of the rotation period). Since the shapes are the same up to a three-dimensional rotation (requiring three Euler angles), and one Euler angle is the rotation around the target's axis, we immediately get the unique pole direction from two Euler angles, and the period from the third rotation angle, when we align the shapes to coincide. Or consider another case: we have not determined the shape, but the geometries have revealed a direction in which the lightcurve is flat. Then it must necessarily be parallel to the rotation axis.
\end{proof}

Of course this is not very helpful in practice where the whole time series depends on the orbital geometries even if we have infinite time. In practice, the spin information is strong since rotation is very constraining. If we have lightcurves from a few geometries not close to each other, the correct spin solution stands out from local minima in the goodness-of-fit comparison. The determination of the spin axis direction is also not affected by the shape/albedo ambiguity.

The number of points is not that crucial if there are sufficiently many geometries: successful inversion (though obviously crude for shape) can be carried out for very sparse photometry. The strong constraints of the model make it possible to solve for the period, spin and shape with sampling frequencies orders of magnitude lower than the usual Nyquist criterion \citep{Kaa:04, Dur.ea:07, Dur.ea:09, Dur.Han:18}. This is the most important modelling mode for asteroids, since useful data can be obtained from hundreds of thousands of targets from large sky surveys, much more than can ever be individually targeted. If the brightness data are so scarce that no models of single targets can be made, one can still infer statistically meaningful distributions of some spin and shape characteristics in large object populations. The crude but analytical and provably unique model is based on the variation levels of extensive sets of brightness measurements \citep{Nor.ea:17, Kaa.Dur:20}. The underlying assumption is that such sets cover the required observing geometries widely and densely enough. The approach is mathematically quite new, potentially introducing a new class of inverse problems. 

Models can also be constructed from combined sets of data obtained at various wavelengths. With thermal infrared and optical data, this can be used to infer surface characteristics from thermal models \citep{Del.ea:15, Dur.ea:17}, and other infrared lightcurves can safely be combined with optical data as the differences between lightcurves at different wavelengths are small enough in view of the obtainable resolution \citep{Dur.ea:18c}.

\subsubsection{Implications for observers}

\paragraph{Spin and period}
As mentioned above, several lightcurves observed under sufficiently different geometries are enough to uniquely reconstruct the shape and spin state. What does it mean in practice? For near-Earth asteroids that come close to Earth, the geometry usually changes a lot during the encounter so it is in principle possible to reconstruct the model from a single apparition. For main-belt asteroids, we usually have to wait several years to collect data from at least two or three apparitions to cover the geometry. 

Another practical issue is the determination of the sidereal rotation period. The local minima in the period parameter space are separated by about $0.5 P^2 / \Delta T$, where $P$ is the rotation period and $\Delta T$ is the interval covered by observations \citep{Kaa:01}. So with any inversion algorithm, we have to make sure that the global minimum is found otherwise the spin and shape solution may be wrong even if the difference between the periods for local and global minimum is small. With a gradient-based approach, this is done by starting the optimization at dense-enough set of initial periods on the whole interval of possible values. The range can be estimated by finding the synodic rotation period and its uncertainty (Fourier analysis). With sparse-in-time photometry, the sampling is sparse with respect to $P$ and we do not have any initial estimate of the period. Then the interval has to cover all physically realistic values -- from the spin-barrier limit of two hours up to hundreds or thousand hours.

\paragraph{Pole $\lambda + 180^\circ$ ambiguity}
In reality, the geometry is not strictly coplanar and data are not ideal, so the shapes corresponding to $\lambda$ and $\lambda + 180^\circ$ poles are not exactly mirrored versions of each other. If such shapes are then used to fit thermal infrared data, one of them can fit significantly better that the other and this can lead to an incorrect conclusion that the IR data resolved the ambiguity. However, this might be just because shapes from lightcurves are slightly different and thermophysical modelling is sensitive to shape details \citep{Han.ea:15}. An example is asteroid (720)~Bohlinia for which \cite{Del.Tan:09} showed that one pole solution is significantly better than the other. However, given the low orbital inclination of only about $2^\circ$, this was caused just by random variations of the shape rather than by breaking the degeneracy.

Also note that $\lambda + 180^\circ$ ambiguity has nothing to do with the prograde/retrograde ambiguity in pole direction, which is usually caused by insufficient data.

\paragraph{Light-scattering model}
In theory, the scattering properties of the surface can be reconstructed together with the shape (Theorem~\ref{th:master}). In practice, the scattering model is rather simple and often fixed. With relative lightcurves, we do not have any information of how the brightness changes with solar phase angle as the phase angle is almost constant for individual lightcurves. With absolute lightcurves or calibrated sparse photometry that cover a wide range of phase angles, it is necessary to fit also the phase curves. This can be done either by using a physical light-scattering model (Hapke or Lumme-Bowell, for example) or a simpler mathematical model -- an exponential-linear phase function multiplied by a combination of Lomell-Seeliger and Lambert scattering \citep{Kaa.ea:01} -- for example. An advantage of this model is that it separates the phase-dependent part and the other part that depends only on illumination and emission angles. So with sparse data or absolute lightcurves, all phase parameters are fitted, while with relative lightcurves, the phase-dependent part is held constant. If the data do not cover low phase angles, the exponential part can be set to zero and only the linear part is fitted.

The phase-angle behaviour is usually assumed to be the same for all data, i.e., the coefficients are the same independently on filter. However, due to the phase-reddening effect, accurate photometry can reveal differences in phase curves for different colors and this requires separate parameters for different colours.

\paragraph{Relative/calibrated/absolute photometry}
By absolute photometry we mean that the brightness is known in absolute flux units. This is not needed in practice, because without the knowledge about the asteroid size, there is always size/reflectance ambiguity that cannot be solved with optical data only. So absolute photometry is needed only when optical photometry is combined with thermal infrared measurements and the model is scaled to real size with a realistic scattering model \citep{Del.ea:15, Dur.ea:17}.

For the purpose of lightcurve inversion, calibrated photometry is sufficient, i.e., for a subset of lightcurves or data points we have information about their common zero point, so they are internally calibrated. If there are more subsets with different zero points, they are arbitrarily shifted on magnitude scale. An example are observations in different filters. Sparse-in-time data typically cover a wide interval of phase angles, they have to be internally calibrated and the phase-function coefficients have to be fitted.

With relative photometry, only the shape of the lightcurve is know, not its brightness with respect to other lightcurves, so there is less information (see Sect.~\ref{sec:implications}). Nevertheless, relative lightcurves are what we usually have to deal with.

\paragraph{YORP}

When talking about spin in Sect.~\ref{sec:spin}, we assumed that asteroid rotation could be described as rotation around a fixed spin axis with a constant rotation rate. However, rotation state of asteroids is affected by solar radiation and anisotropic thermal radiation from the surface. If the net torque does not cancel out when averaged over the surface, rotation, and orbit, then it leads to a secular change of the rotation period and the spin axis direction -- an effect know as Yarkovsky--O'Keefe--Radzievskii--Paddack (YORP) effect. The YORP-induced change of rotation rate can be directly observed as a change of the rotation period \citep{Low.ea:07, Tay.ea:07} or as a change of the rotation phase accumulated over long time \citep{Kaa.ea:07}. The YORP effect can be easily included into the model as an additional parameter describing the linear change of the rotation frequency with time. YORP also changes the direction of spin axis but these changes are so slow that they cannot be detected from current data.

\paragraph{Tumbling}

Even without any external force, the rotation of an asteroid can be more complicated than a relaxed rotation around a fixed axis. Some asteroids are in force-free precessing rotation state. Such rotation can be described by Euler equations that determine the evolution of Euler angles. The rotation is fully described by eight parameters: the angular momentum vector (three components), initial values of three Euler angles, and two principal moments of inertia (normalized by the third one). The size of the angular momentum and the initial nutation angle can be replaced by two other parameters: the rotation and precession periods that can be estimated from Fourier analysis of lightcurves \citep{Kaa:01, Pra.ea:05}. In practice, reconstruction of the tumbling rotation state is more complicated than in case of principal-axis rotation, scanning of all physically possible combinations of parameters is time-consuming, so the best way how to proceed is to estimate candidate rotation and precession periods from the most prominent periods in the lightcurve signal \citep{Sch.ea:10,Pra.ea:14}.

\paragraph{Binaries}

Binary and multiple systems are in general more ``interesting'' than single bodies \cite[see the review of][for example]{Marg.ea:15} but reconstructing their geometrical and dynamical parameters is much more complicated. A general case can be very complex: two similar-size bodies orbiting each other on an eccentric orbit, their rotation periods not synchronous with the orbital period, the orbital plane may be precessing on longer time scales, the spins may not be aligned with the normal to the orbital plane, etc. The most detailed model we have so far was reconstructed by \cite{Ost.ea:06} from radar observations of asteroid 1999~KW4. The system was later  directly imaged by VLT.\footnote{\url{https://www.eso.org/public/news/eso1910/}} 

We cannot expect to reconstruct such details from lightcurves. However, mutual events that we observe in lightcurves can be fitted with a simple model assuming ellipsoidal shapes of the components and geometry of the orbit can be reconstructed \citep{Sch.Pra:09, Car.ea:15}. 

A special case is a fully synchronous system, i.e., when two bodies orbit each other on a circular orbit and both rotation periods are the same as the orbital period. For a distant observer, such system behaves like a single body because the mutual orientation of the components does not change. A slight modification of algorithms that can model nonconvex shapes enables us to model synchronous binaries. However, the nonconvex ambiguity applies also here, so we cannot expect to reconstruct correct shapes from lightcurves only. Any disk-resolved data would help, of course, so for example, two stellar occultations by asteroid (90)~Antiope revealed clear concavities \citep[][Sect.~3.10]{Dur.ea:15b}.

\subsection{Master theorem}

Taking into account all of the above, we can, for completeness, write the full theorem of lightcurve inversion in all its glory. It includes the solution of the shape, spin, and physically realistic scattering behaviour (not explicitly stated previously):

\begin{theorem} 	
 \label{th:master}
 Extensive brightness data at nonzero solar phase angles uniquely determine the shape, spin, and scattering properties (excluding albedo variegation) of an asteroid (when the shape is constrained as in \cite{Kaa.ea:92} and \cite{Vii.ea:17}, and scattering as in \cite{Kaa.ea:92} for a unique solution).
\end{theorem}

\begin{proof}
 The uniqueness and stability proof for convex shapes is summarised in, e.g., \cite{Kaa.Lam:06}. In \cite{Vii.ea:17} and \cite{Kaa:19a}, it was shown that, with some assumptions, a larger class of shapes than convex can be called reconstructable as suggested by numerical tests; e.g., \cite{Dur.Kaa:03, Kaa.Dur:07}. However, it is certainly a smaller class than that of tangent-covered bodies; i.e., the shapes reconstructable from silhouettes (of which convex bodies are a subset), and the reconstruction is not stable. (In practice: if a convex solution and a nonconvex one fit the data as well, there are infinitely many nonconvex shapes between the two that also fit the data.) The spin solution is discussed above. For scattering, a general form of a scattering model is given in eq.~(2.6) of \cite{Kaa.Dur:07}, and its truncations such as eq.~(2.7) are very close to typical scattering models such as Lommel-Seeliger, Lambert, or Hapke. The form (2.6) is exactly the same as eq.~(5.15) in \cite{Kaa.ea:92}, which leads to a unique solution of the scattering coefficients via eq.~(5.16). This is because the higher the order (the resolution level) of an unknown shape parameter is, the more there are equations for it, and this extra information can be used for a scattering solution.
\end{proof}

\begin{remark}
 The Minkowski stability for the (convex) shape means that rather different scattering properties can account for the data with quite similar shapes; thus, the scattering solution is unstable so that one uses essentially fixed scattering in inversion.
\end{remark}

\subsection{Instabilities near nonuniqueness; summary}

We conclude our review of the uniqueness and stability properties of the photometric inverse problem by discussing some cases where Minkowski stability does not help us. In inverse problems of continuously parametrised systems, ambiguities also translate to instabilities. If there is a nonuniqueness under some conditions, then there are unstable solutions under nearby conditions.  For example, even the zero solar phase angle is actually sufficient for a unique shape solution if the scattering is suitable \citep{Kaa.ea:92}. The often used combination of Lommel-Seeliger and Lambert scattering is suitable in principle: the L-S part (equivalent to geometric scattering at zero solar phase) yields the shape coefficients of even degree, and the Lambert part (diffuse scattering) those of odd degree. Scattering with a specularity exponent between one (L-S) and two (Lambert) would also yield all the coefficients. But the solution is inevitably unstable because of the data and model noise: any uncertainty of the brightnesses or the Lambert component, close to zero for dark bodies, results in large uncertainty in the coefficients of odd degree. The Lambert component is useful only at nonzero phase angles when its uncertainty is not a big issue. But, again, the phase angle needs to be sufficiently large (more than, say, ten degrees) for us to be clear from the instability region around the nonuniqueness condition. A further example of the same phenomenon is the case of nonconvexities close to strictly concave ones (those without saddle-like surface parts). Concavities, no matter how large or how deep, are never uniquely reconstructable from photometry alone \citep{Kaa.ea:04, Vii.ea:17} (the latter also contains further ambiguity-instability examples), so the reconstruction of any nonconvexities close to the nonuniqueness condition is unavoidably unstable. 

In summary, to obtain a fast global-scale model of an asteroid from photometry, one can solve for the spin and construct a convex shape with Convexinv.\footnote{\url{https://astro.troja.mff.cuni.cz/projects/damit/files/version_0.2.1.tar.gz}} For further educated guesses, one can construct photometry-based nonconvex models with ADAM\footnote{\url{https://github.com/matvii/ADAM}} using both octantoid and subdivision surface representations (see Sect.~\ref{sec:disk-resolved_data} below) if there are several dense lightcurves to warrant this. A number of solutions obtained with various initial conditions are necessary. Then one should compare the results from the three shape representations to identify common large-scale features \citep{Vii.ea:15a}. Such models are not rigorously defined and they are mostly governed by the representation type and the chosen regularisation functions, but they may be useful for giving an impression of what the target might look like \citep[cf. the nonconvex models in][]{Kaa.ea:01, Kaa.ea:04}. 

Unless large solar phase angles are available \citep{Dur.Kaa:03}, the convex solution is the only scientific one in the sense that it has well-defined acceptance criteria and its approximate nature is well understood. It is the simplest solution with the fewest assumptions sufficient for explaining the data down to the noise level. The acceptance of a more elaborate model can only be based on manifestly better data fit (requiring large phase angles) or special assumptions limited to cases where the convex solution is not feasible even as an approximation \citep[e.g., choosing a bi-lobed shape instead of a convex option, cf.][]{Dur.Kaa:03}.
Otherwise such a model is not likely to be closer to the correct shape than the convex solution by any quantitative criteria \citep[see also the discussion at the end of the Appendix of][]{Kaa:19a}.

\section{Snapshot images, occultations, multidata}
\label{sec:disk-resolved_data}

The GP of adaptive optics (AO) images is essentially the photographic projection, with instrumental point-spread functions. Occultations basically provide samples of the silhouette of the target. Since AO processing typically causes some artefacts inside the image, and in any case the strongest contrast is along the contours of occlusion or shadows, most of the information is in the boundary contours \citep{Fet.ea:19}. In the same way as the inversion of the brightness GP data can be proven to be unique for convex bodies (and conditionally to some nonconvex ones; \cite{Vii.ea:17, Kaa:19a}), a class of nonconvex bodies can be proven to be reconstructable from image boundary contours \citep{Kaa:11}. This class is larger than that of tangent-covered bodies.

Since most of the information is in the dark/light boundary contours, one can use these directly in shape reconstruction (Kaasalainen 2011); this also removes any modelling errors due to scattering from the image analysis. On the other hand, one can also use the images as such: data or model errors in the images inside the contours are not very significant as weighed against the contour pixels. A robust way of doing this is to perform the fit in the Fourier transform domain rather than pixels \citep{Vii.Kaa:14, Vii.ea:15a}: this makes gradient-based optimisation and the control of the size of image features especially convenient.

Often there are not enough images for complete reconstruction; in such cases we can combine all the data from various instruments, especially lightcurves, with the images in the inversion process. This is in principle straightforward; the main question is the weight assigned to each data type. Finding the proper weights requires some experimenting and sampling. A criterion for the optimal weights can be established by requiring that each mode contribute its essential information but not more to yield solutions as consistent with each data source as possible. This is called the maximum compatibility estimate \citep{Kaa:11}. A procedure for combining all available data sources to produce asteroid models, All-Data Asteroid Modelling (ADAM), is described in \cite{Vii.ea:15a}. ADAM can be used for any data source separately (including nonconvex models from lightcurves only, but note the warnings above) or combined with others. 

Several asteroid models have been constructed especially with AO data and lightcurves, using either ADAM (the Fourier transform mode) or extracted boundary contours \citep{Car.ea:10, Car.ea:12, Mer.ea:13, Vii.ea:15b, Vii.ea:17, Han.ea:17a, Han.ea:17b}. Flyby images have also been combined with lightcurves in an analogous way to produce as constrained models as possible of the ``dark side'' combined with the exact cartography of the seen side \citep{Kel.ea:10, Sie.ea:11}. Occultation data can be included in the inversion process as such or used as a consistency check \citep{Dur.ea:11}. With ever-improving AO equipment, the best images in the process are approaching cartographic resolution \citep{Vii.ea:18, Car.ea:19, Ver.ea:18, Ver.ea:20, Mar.ea:20}, almost carrying the problem away from the realm of inverse problems into cartography, photoclinometry \citep{Gas.ea:08}, etc.\ that have well-established methods. But with all AO images, it is important to remember that a feature seen in an image can only be considered real and verified if it is seen in at least another image in a consistent manner. Otherwise it is only plausible or, if inconsistent with other data, nonexistent. Such spurious detail can be caused by the deconvolution or post-processing of raw AO images. The deconvolution of images is a separate inverse problem often practical for shape modelling and feature identification. While not perfect, it can be seen as further 3D-model-independent regularisation that introduces useful prior constraints (in addition to the point-spread function) into the full inverse problem.  

Another essential concept, closely linked with shape class, is the shape representation. There simply is no universal way of representing shape models in inverse problems: the choice always restricts the convergence and possible shape solutions in some way.\footnote{Add to this the data noise, restricted observation geometries, partly subjective choices in regularisation, and systematic errors in both data and the model (scattering etc.), and it becomes obvious that the whole thing is not solvable just by plugging some components together and letting the machine run.}  Convex bodies are simple to describe by the positive curvature translatable to shape via the Minkowski procedure, whereas nonconvex shapes have various subclasses and thus representations. These can be, e.g., starlike or radius displacements from given basic shapes or something specially made for a problem case \citep{Kaa.ea:01, Ost.ea:02b}, or generalised starlike/octantoid or subdivision control points \citep{Kaa.Vii:12, Vii.Kaa:14, Bar.Dud:18}, or level sets, etc. Ideally, at least two shape representations should be used in any given problem to establish the role of systematic shape effects \citep{Vii.ea:15a}.

An important issue in this context is also ``{\em inverse crime}'' \citep{Kai.Som:05}. Simulating data is necessary in testing an inversion method (and the general uniqueness and stability of the problem) or the information content available from a particular observational setup \citep[e.g.,][]{Nor.ea:17}. In such cases, one easily makes the mistake of creating the data with the same model they are inverted with, which is called inverse crime. This can result in much too overconfident solutions and exaggerated resolution. Avoiding inverse crime is not just about not using the same discretisation level in data creation and analysis: it applies to data noise (that should include non-Gaussian elements), surface scattering (that, in addition to different forward and inverse problem models, should include some random behaviour over the surface in data creation), the surface representations, and so on. Even with these precautions, computer simulations in asteroid modelling are always likely to produce better results than in real life simply because nature is much more imaginative in producing variation in physical conditions and error sources than we are. 

In spite of the speed and parallelisability of computers and clusters, some numerical and computational aspects require consideration. In optimisation, procedures that only use the function value are slow and converge poorly. Genetic algorithms are a particularly poor variation of these for this problem, and they are not well suited for general analysis as they typically lose the information on other eligible minima. In this problem class, local minima are usually sparse enough so that gradient-based optimisation procedures (in particular Levenberg-Marquardt designed for chi-square-type problems such as these) started from a variety of initial guesses are superior in both speed and accuracy. Moreover, by exploring the whole parameter space we can be sure that the global minimum is not missed. Also, the best efficiency by far is obtained from analytical gradients readily available in this case. In integration over model surface, ray-tracing is fast to carry out with some pre-listing of occluding surface patches \citep[triangles;][]{Kaa.ea:01}. Many parts of the problem parallelise easily to, e.g., GPUs \citep{Vii.ea:15a}. One can also use game-based image processing hardware, but this may not be useful for obtaining the function gradients. 

Obviously, no single rendering of the solution is sufficient in asteroid modelling from any sources: the differences between the rendering samples help to understand which parts of the solution are best constrained by the data. On the other hand, the usual sampling methods (MCMC) are slow and cannot be employed because systematic and model errors dominate over noise and are not modellable in the usual MCMC fashion, and because the chosen shape representation biases the results. The only bias-free shape space to sample would be a large set of mesh vertices each free to move around, but this is extremely difficult to arrange in practice, and even then one would need regularisation to keep the mesh from getting tangled. After various experiments, we have found the simplest way to sample is to compare the best solutions from different setups. 

\subsection{Implications for observers}

\paragraph{Adaptive optics}
Although AO images taken with the VLT/SPHERE instrument look almost like photographs taken by spacecraft \citep[Hygiea in][for example]{Ver.ea:20}, they are always product of \emph{deconvolution}, which can be source of inconsistencies and artefacts. For example, the deconvolved images of (4)~Vesta in \cite{Fet.ea:19} reveal amazing details on the surface at the diffraction limit (the limb was reconstructed with resolution better that the formal diffraction limit) but still there are artefacts on the disk near the limb. In general, the result of deconvolution depend on the method, parameter setup, and there is always a trade-off between the level of details and presence of false artefacts.

\paragraph{Occultations} 

With a dramatic increase of accuracy of stellar catalogues due to Gaia astrometry, occultations by asteroids can be now predicted with higher accuracy. The technical equipment is also improving so occultations will play an important role in shape modelling. They are also the only way how to obtain information about shapes of trans-Neptunian objects that are too far to be reconstructable by lightcurve inversion (opposition geometry) and also unresolvable by AO.

\section{Stranger projections: radar, interferometry}

Interferometry, especially that from ALMA at thermal IR wavelengths, is another high-resolution data source of asteroids. It is essentially the Fourier transform of a (thermal) image, so ADAM is readily usable for inversion \citep{Vii.ea:15a, Vii.ea:15b}. For the available resolution level (analogously to low-resolution AO images, boundary contours provide most of the shape information), the thermal modelling can be done fast by a semianalytical Fourier-series model \citep{Nes.Vok:08}.

Radar is the main source of high-resolution information from near-Earth asteroids \citep{Ben.ea:15}. Its GP is somewhat strange in the sense that it projects surface points with a given radial velocity and distance from the radar onto one projection image point. Thus the image is a many-to-one mapping, contrary to usual images, which affects its information content. While computational inversion of range-Doppler radar data has been in use since the early 1990's as described in \cite{Ost.ea:02b}, the uniqueness of the shape solution of the inverse problem was actually proved in \cite{Vii.Kaa:14}. The proof was based on using the boundary contours of range-Doppler images (that are even more pronounced in their information content than those of AO images), and it was shown that the shape class reconstructable from these is that of tangent-covered bodies. Other parts of the radar image provide further information and enlarge the reconstructable class. In \cite{Kaa.Lam:06}, it was shown that mere range radar data from a spherical target uniquely determine not only the reflectivity (corresponding to albedo map) on its surface, but also the distribution of other radar scattering properties. Doppler radar data uniquely determine the reflectivity.

Radar is thus a provably rich source of information, and it has the potential for very high resolution. However, it should be remembered that, due to the inherent ambiguities of the range-Doppler plots as well as the modelling and data errors, the actual resolution level of a radar-based model is seldom as high as the nominal (apparent) resolution of the images. This is well apparent from the ground truth probe images of Itokawa and Toutatis (although the available radar observation geometries were somewhat limited), or by the difference of the features of the Kleopatra radar model model with those of the AO one (even though the models are qualitatively similar). On the other hand, the case of asteroid Bennu, essentially a convex shape covered by many small boulders and a large one, shows that some detailed features (including the large boulder on the southern hemisphere) can be captured by a tailored radar model \citep{Nol.ea:13, Nol.ea:19}.

There are various software packages for radar modelling; radar data can be automatically included in ADAM input, and ADAM is orders of magnitude faster than traditional radar codes. In any case, the addition of special features to the ADAM or any other basic model usually requires a hand-crafted model within a model. Radar data are peculiar in that it is possible to have a sequence of high-resolution radar images without being able to model and place the corresponding detail anywhere in a 3D-model with reasonable certainty. With the ordinary photographic image projection, details from even one high-resolution image can in principle always be added (at least manually and tentatively) to a low-resolution shape model mostly based on other data such as lightcurves. Thus, assessing the reliability of shape models from radar data especially at limited geometries is considerably more complicated than in the case of ordinary plane-of-sky projections.

\setlength{\bibhang}{1em}



\newcommand{\SortNoop}[1]{}

\end{document}